\documentclass[11pt,draftcls,onecolumn]{IEEEtran}
\usepackage{epsfig,graphics,subfigure,amsthm,psfrag,amsmath,amssymb,cite,citesort}
\usepackage{amsmath,amssymb,eucal}
\usepackage{xifthen}
\usepackage{mathtools}
\usepackage{enumerate}
\usepackage{microtype}
\usepackage{xspace}
\usepackage{bm}
\usepackage[T1]{fontenc}
\usepackage{fancyhdr}
\usepackage{lastpage}
\usepackage[center]{caption}

\newcommand{\mbs}[1]{\bm{#1}}
\newcommand{\vect}[1]{{\lowercase{\mbs{#1}}}}
\newcommand{\mat}[1]{{\uppercase{\mbs{#1}}}}

\renewcommand{\Bmatrix}[1]{\begin{bmatrix}#1\end{bmatrix}}

\newcommand{\T}{{\scriptscriptstyle\mathsf{T}}}
\renewcommand{\H}{{\scriptscriptstyle\mathsf{H}}}

\renewcommand{\Re}[1][]{\ifthenelse{\isempty{#1}}{\operatorname{Re}}{\operatorname{Re}\left(#1\right)}}
\renewcommand{\Im}[1][]{\ifthenelse{\isempty{#1}}{\operatorname{Im}}{\operatorname{Im}\left(#1\right)}}

\newcommand{\bv}{\vect{b}}
\newcommand{\cv}{\vect{c}}

\newcommand{\ev}{\vect{e}}

\newcommand{\hv}{\vect{h}}

\newcommand{\nv}{\vect{n}}

\newcommand{\uv}{\vect{u}}
\newcommand{\vv}{\vect{v}}
\newcommand{\wv}{\vect{w}}
\newcommand{\xv}{\vect{x}}
\newcommand{\yv}{\vect{y}}
\newcommand{\zv}{\vect{z}}

\newcommand{\etav}{\vect{\eta}}

\newcommand{\Am}{\mat{a}}
\newcommand{\Bm}{\mat{b}}
\newcommand{\Cm}{\mat{c}}

\newcommand{\Hm}{\mat{h}}

\newcommand{\Km}{\mat{k}}

\newcommand{\Qm}{\mat{q}}

\newcommand{\Sm}{\mat{s}}

\newcommand{\Um}{\mat{u}}
\newcommand{\Vm}{\mat{V}}
\newcommand{\Wm}{\mat{w}}

\newcommand{\Cc}{{\mathcal C}}
\newcommand{\Dc}{{\mathcal D}}

\newcommand{\Mc}{{\mathcal M}}
\newcommand{\Nc}{{\mathcal N}}

\newcommand{\Rc}{{\mathcal R}}

\newcommand{\Tc}{{\mathcal T}}
\newcommand{\Uc}{{\mathcal U}}

\newcommand{\Xc}{{\mathcal X}}

\newcommand{\CC}{\mathbb{C}}

\newcommand{\Id}{\mat{\mathrm{I}}}

\newcommand{\CN}[1][]{\ifthenelse{\isempty{#1}}{\mathcal{N}_{\mathbb{C}}}{\mathcal{N}_{\mathbb{C}}\left(#1\right)}}

\renewcommand{\P}[1][]{\ifthenelse{\isempty{#1}}{\mathbb{P}}{\mathbb{P}\left(#1\right)}}
\newcommand{\E}[1][]{\ifthenelse{\isempty{#1}}{\mathbb{E}}{\mathbb{E}\left(#1\right)}}
\renewcommand{\det}[1][]{\ifthenelse{\isempty{#1}}{\text{det}}{\text{det}\left(#1\right)}}
\newcommand{\trace}[1][]{\ifthenelse{\isempty{#1}}{\text{tr}}{\text{tr}\left(#1\right)}}
\newcommand{\rank}[1][]{\ifthenelse{\isempty{#1}}{\text{rank}}{\text{rank}\left(#1\right)}}
\newcommand{\diag}[1][]{\ifthenelse{\isempty{#1}}{\text{diag}}{\text{diag}\left(#1\right)}}

\DeclarePairedDelimiter\abs{\lvert}{\rvert}
\DeclarePairedDelimiter\Abs{\lvert}{\rvert^2}
\DeclarePairedDelimiter\norm{\lVert}{\rVert}
\DeclarePairedDelimiter\Norm{\lVert}{\rVert^2}
\DeclarePairedDelimiter\normf{\lVert}{\rVert_{\text{F}}}
\DeclarePairedDelimiter\Normf{\lVert}{\rVert^2_{\text{F}}}

\newcommand{\defeq}{\triangleq}
\newcommand{\eqdef}{\triangleq}

\newtheorem{definition}{Definition}
\newtheorem{theorem}{Theorem}
\newtheorem{lemma}{Lemma}
\newtheorem{assumption}{Assumption}

\begin{document}

\title{The DoF Region of the Multiple-Antenna \\Time Correlated Interference Channel \\ with Delayed CSIT}
\author{\authorblockN{Xinping Yi, \emph{Student Member, IEEE}, David Gesbert, \emph{Fellow, IEEE},\\  Sheng Yang, \emph{Member, IEEE},  Mari Kobayashi, \emph{Member, IEEE}}
\thanks{X.~Yi and D. Gesbert are with the Mobile Communications Dept., EURECOM, 06560 Sophia Antipolis, France (email: \{xinping.yi, david.gesbert\}@eurecom.fr).}
\thanks{S.~Yang and M. Kobayashi are with the Telecommunications Dept., SUPELEC, 91190 Gif-sur-Yvette, France (e-mail:
\{sheng.yang, mari.kobayashi\}@supelec.fr).}
\thanks{Part of this paper is to be presented in \cite{YiGe:2012}}
}

\maketitle

\begin{abstract}
 We consider the time-correlated multiple-antenna interference channel where the transmitters have (i) {\em delayed} channel state information (CSI) obtained from a latency-prone feedback channel as well as (ii) imperfect {\em current} CSIT, obtained e.g. from prediction on the basis of these past channel samples. We derive the degrees of freedom (DoF) region for the two-user multiple-antenna interference channel under such conditions. The proposed DoF achieving scheme exploits a particular combination of the space-time alignment protocol designed for fully outdated CSIT feedback channels (initially developed for the broadcast channel by Maddah-Ali et al, later extended to the interference channel by Vaze et al. and Ghasemi et al.)  together with the use of simple zero-forcing (ZF) precoders. The essential ingredient lies in the quantization and feedback of the residual interference left after the application of the initial imperfect ZF precoder. Our focus is on the MISO setting albeit extensions to certain MIMO cases are also considered.
\end{abstract}
\section{Introduction}
Although the determination of capacity region of the interference channel (IC) has been a long standing open problem, several interesting recent results shed light on the problem from various perspectives. Among these we may cite the capacity region obtained for special cases \cite{Han:81,Sato:81,Motahari:09,Shang:09}, or obtained for general channel classes in both scalar and MIMO settings up to approximations with bounded gaps \cite{Etkin:08,Karmakar:11}.  When specializing to the large SNR regime, it is known that the characterization of the full capacity region can be conveniently replaced with the determination of the so-called degree-of-freedom (DoF) region. Progress on that particular front was reported in \cite{Jafar:2007} with the derivation of the DoF region for the two-user MIMO interference channel with $M_1$, $M_2$ transmit antennas and $N_1$, $N_2$ receive antennas, where the sum DoF $\min\{M_1+M_2, N_1+N_2,\max(M_1,N_2),\max(M_2,N_1)\}$ is shown to be optimal. Most of these advances suggest achievable schemes which require the full knowledge of channel state information (CSI) at both the transmitter and receiver sides. In fact, the cruciality of CSI at the transmitter side in particular is demonstrated in such works as \cite{Huang:2012,Zhu:2012,Vaze:2009} where the DoF region is shown to shrink dramatically when zero CSIT is available.  The intermediate scenario of {\em limited} or {\em incomplete}  CSIT was also considered in \cite{Bolcskei:2009,Krishnamachari:2010}.  In \cite{Bolcskei:2009}, the rate of limited feedback needed to preserve the DoF optimality in interference alignment-enabled IC is provided. More recently, the impact of feedback delays providing the transmitter with outdated CSI over MIMO channels was considered in \cite{Maddah-Ali:10} for the broadcast channel (BC) and later extended to the IC \cite{Vaze:2011,Ghasemi:2011}.  The key contribution in \cite{Maddah-Ali:10} was to establish the usefulness of even completely outdated channel state information in designing precoders achieving significantly better DoF than what is obtained without any CSIT. Considering the worst case scenarios, including those where the feedback delay extends beyond the coherence period of the time varying fading channels, the authors in~\cite{Maddah-Ali:10} propose a space-time interference alignment-inspired strategy achieving an optimal sum DoF of 4/3 for the two-user MISO BC, in a setting when the no CSIT case yields no more than 1 DoF. The essential ingredient for the proposed scheme in \cite{Maddah-Ali:10} lies in the use of multi-slot protocol initiating with the transmission of unprecoded information symbols to the user terminals, followed by the analog forwarding of the interference overheard in the first time slot.

Recently, this strategy was generalized under similar principle to the interference channel setting \cite{Vaze:2011,Ghasemi:2011}, again establishing DoF strictly beyond the ones obtained without CSIT  in scenarios where the delayed CSIT bears no correlation with the current channel realization.

Albeit inspiring and fascinating in nature, such results nonetheless rely on the somewhat over-pessimistic assumption that no estimate for the  {\em current} channel realization is available to the transmitter. Owing to the finite Doppler spread behavior of fading channels, it is however the case in many real life situations that the past channel realizations can provide information about the current one. Therefore a scenario where the transmitter is endowed with delayed CSI in addition to some (albeit imperfect) estimate of the current channel is practical relevance. This form of combined delayed and imperfect current CSIT was recently introduced in \cite{Kobayashi:2012} for the multiple-antenna broadcast channel whereby a novel transmission scheme is proposed which extends beyond the MAT algorithm in allowing the exploitation of precoders designed based on the current CSIT estimate. The full characterization of the optimal DoF for the hybrid CSIT was reported in \cite{Yang:2012} and independently in \cite{Gou:2012}.  The key idea behind
the schemes in \cite{Kobayashi:2012,Yang:2012} lies in the modification of the MAT protocol where i) the initial time slot involves transmission of {\em precoded} symbols, followed by the forwarding of {\em residual} interference overheard in the first time slot, and ii) the taking advantage of the reduced power for the residual interference (compared with full power interference in MAT) based on a suitable quantization method and digital transmission.

In this paper, we extend the results in~\cite{Kobayashi:2012,Yang:2012} and consider the two-user time-correlated multiple-antenna interference channel. A similar hybrid CSIT scenario is considered whereby each transmitter has access to delayed channel samples for the links it is connected to, as well as possessing an imperfect estimate of the current channel. The current CSIT estimate could be obtained from, e.g., a linear prediction applied to past samples \cite{Lapidoth:2005,Caire:2010}, although the prediction aspects are not specified in this paper.
Instead, the quality level for the current CSIT estimate is simply modeled in terms of an exponent of the transmit power level, allowing DoF characterization for various ranges of current CSIT quality. Thus our model bridges between previously reported CSIT scenarios such as the pure delayed CSIT of \cite{Maddah-Ali:10,Vaze:2011,Ghasemi:2011} and the pure instantaneous CSIT scenarios. We assume each receiver has access to its own perfect instantaneous CSI and the perfect delayed CSI of other receivers (as in e.g. \cite{Maddah-Ali:10,Vaze:2011,Ghasemi:2011}), in addition to the imperfect current CSI.

In what follows we obtain the following key results:
\begin{itemize}
  \item We establish an outer bound on the DoF region for the two-user temporally-correlated MISO interference channel with perfect delayed and imperfect current CSIT, as a function of the current CSIT quality exponent. This result is initially derived for the two-antenna transmitters and then generalized.
  \item We propose two schemes which achieve the key vertices of the outer bound with perfect delayed and imperfect current CSIT. The schemes build on the principles of time-slotted protocol, starting with the ZF precoded transmission of information symbols from the two interfering transmitters simultaneously and followed by forwarding of the residual interferences. As in the BC case, the residual interference reflects on the quality of the initial precoder and can be shown to be quantized and power scaled in a suitable way to achieve the optimal DoF.
	\item Our results coincide with previously reported DoF results for the perfect CSIT setting (current CSIT of perfect quality) and pure delayed CSIT setting (current CSIT of zero quality).
    \item The DoF region of certain MIMO cases is also provided as a function of the current CSIT quality exponent and the number of receive antennas.
\end{itemize}

\textbf{Notation}: Matrices and vectors are represented as uppercase and lowercase letters, and matrix transport, Hermitian transport, inverse and determinant are denoted by $\Am^\T$, $\Am^\H$, $\Am^{-1}$ and $\det(\Am)$, respectively. $\hv^{\bot}$ is the normalized orthogonal component of any nonzero vector $\hv$. The approximation $f(P) \sim g(P)$ is in the sense of $\lim_{P \to \infty} \frac{f(P)}{g(P)}=C$, where $C$ is a constant that does not scale as $P$. $\Am \succeq 0$ means $\Am$ is symmetric positive semidefinite if $\Am$ is square and $\Am \preceq \Bm$ means $\Bm-\Am$ is symmetric positive semidefinite if both $\Am$ and $\Bm$ are squared matrices.

\section{System Model}
We consider a two-user MISO interference channel, where two transmitters each equipped with $2$ antennas\footnote{The generalization to arbitrary number of antennas is considered in Section VI.} wish to send two private messages to their respective receivers each with a single antenna, as shown in Fig.~1. The discrete time baseband signal model is given by
\begin{subequations}
\begin{align}
y(t) &= \hv_{11}^\H(t) \xv_1(t) + \hv_{12}^\H(t) \xv_2(t) + e(t) \\
z(t) &= \hv_{21}^\H(t) \xv_1(t) + \hv_{22}^\H(t) \xv_2(t) + b(t),
\end{align}
\end{subequations}
for any time instant $t$, where $\hv_{ji}(t) \in \CC^{2\times 1}$ is the channel vector from Tx-$i$ to Rx-$j$; $e(t), b(t) \sim \CN[0,1]$ are normalized additive white Gaussian noise~(AWGN) at the respective receivers; the coded input signal $\xv_i(t)$ is subject to the power constraint $\E( \norm{\xv_i(t)}^2 ) \le P$, $\forall\,t$.

\begin{figure}[htb]
\centering
\includegraphics[width=0.4\columnwidth]{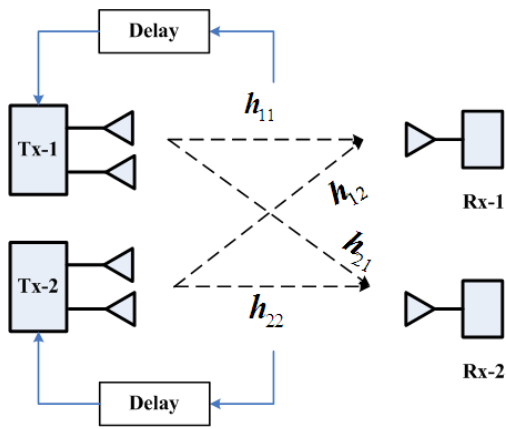}
\caption{The two-user MISO interference channel.}
\label{fig:DoF}
\end{figure}

\begin{assumption} [mutually independent fading]
At any given time instant $t$, the channel vectors $\{\hv_{ji}(t)\}$ are mutually independent and identically distributed (i.i.d.) with zero mean and covariance matrix $\Id_2$.
\end{assumption}

\begin{assumption} [perfect delayed local CSIT and imperfect current local CSIT]
  At each time instant $t$, Tx-$i$ knows perfectly the delayed local CSIT $\{{\hv}_{1i}(k),{\hv}_{2i}(k),k=1,\dots,t-1\}$ (with which link it is respectively connected), and somehow predict/estimate imperfectly the current local CSIT $\{\hat{\hv}_{1i}(t),\hat{\hv}_{2i}(t)\}$, which can be modeled by
  \begin{align}
    \hv_{ji}(t) = \hat{\hv}_{ji}(t) + \tilde{\hv}_{ji}(t)
  \end{align}
  where the estimate $\hat{\hv}_{ji}(t)$ and estimation error $\tilde{\hv}_{ji}(t)$ are independent and  assumed to be zero-mean and with variance $(1-\sigma^2)\Id_2$, $\sigma^2 \Id_2$, respectively ($0 \le \sigma^2 \le 1$).
\end{assumption}

\begin{assumption} [perfect delayed CSIR, imperfect current CSIR and perfect current local CSIR]
  At each time instant $t$, Rx-$i$ knows perfectly the delayed CSIR up to instant $t-1$ for all links, i.e., $\{{\Hm}(k)\}_{k=1}^{t-1}$, where
  \begin{align}
  {\Hm}(k) \defeq \{{\hv}_{11}(k),{\hv}_{12}(k),{\hv}_{21}(k),{\hv}_{22}(k)\},
   \end{align}
   and the imperfect current CSIR (similarly modeled as at the transmitters) up to instant $t$ for all links, i.e., $\{\hat{\Hm}(k)\}_{k=1}^{t}$, where
   \begin{align}
   \hat{\Hm}(k) \defeq \{\hat{\hv}_{11}(k),\hat{\hv}_{12}(k),\hat{\hv}_{21}(k),\hat{\hv}_{22}(k)\},
   \end{align}
    as well as the perfect current local CSIR, i.e., $\{{\hv}_{i1}(t),{\hv}_{i2}(t)\}$.
\end{assumption}

We assume that the estimation error $\sigma^2$ can be parameterized as an exponential function of the power $P$, so that we hope to characterize the DoF of the MISO IC with respect to this exponent. To this end, we introduce a parameter $\alpha \ge 0$, such that
\begin{align}
  \alpha \defeq  -\lim_{P \to \infty} \frac {\log \sigma^2}{\log P}. \label{eq:alpha-def}
\end{align}
This $\alpha$ indicates the quality of current CSIT at high SNR. While $\alpha=0$ reflects the case with no current CSIT, $\alpha \to \infty$ corresponds to that with perfect instantaneous CSIT. As a matter of fact, when $\alpha \ge 1$, the quality of the imperfect current CSIT is sufficient to avoid the DoF loss, and ZF precoding with this imperfect CSIT is able to achieve the maximum DoF~\cite{Caire:2010}. Therefore, we focus on the case $\alpha \in [0,1]$ hereafter. The connections between the above model and the linear prediction over existing time-correlated channel models with prescribed user mobility are highlighted in \cite{Kobayashi:2012}.

According to the definition of the estimated current CSIT, we have
   \begin{align}
     \E (|{\hv}_{ji}^{\H}(t) \hat{\hv}_{ji}^{\perp}(t)|^2)&=\E (|\hat{\hv}_{ji}^{\H}(t) \hat{\hv}_{ji}^{\perp}(t)|^2) + \E (|\tilde{\hv}_{ji}^{\H}(t) \hat{\hv}_{ji}^{\perp}(t)|^2)\\
     &=\E (\abs{\tilde{\hv}_{ji}^\H(t)\tilde{\hv}_{ji}(t)})\\
     &= \sigma^2 \\
     & \sim P^{-\alpha}  \label{eq:P-to-alpha}
   \end{align}
where $\E (\Norm{\hat{\hv}_{ji}^{\perp}(t)})=1$, and (\ref{eq:P-to-alpha}) is obtained from (\ref{eq:alpha-def}).


\section{The Degree of Freedom Region}
A rate pair $(R_1,R_2)$ is said to be achievable for the two-user interference channel with perfect delayed CSIT and imperfect current CSIT if there exists a $\left(2^{nR_1},2^{nR_2},n\right)$ code scheme consists of:
\begin{itemize}
  \item two message sets $[1:2^{nR_1}]$ at the Tx-1 and $[1:2^{nR_2}]$ at the Tx-2, from which two independent messages $\Mc_1$ and $\Mc_2$ intended respectively to the Rx-1 and Rx-2 are uniformly chosen;
  \item one encoding function at the Tx-$i$:
    \begin{align}
    \xv_i(t) &= f_{i} \left(\Mc_i,\{{\hv}_{1i}(k)\}_{k=1}^{t-1},\{{\hv}_{2i}(k)\}_{k=1}^{t-1},\{\hat{\hv}_{1i}(k)\}_{k=1}^t,\{\hat{\hv}_{2i}(k)\}_{k=1}^t\right) \label{eq:enc-fun};
    \end{align}
  \item and one decoding function at its corresponding receiver, e.g.,
    \begin{align}
    \hat{\Mc}_j &= g_{j} \left(\{y(t)\}_{t=1}^{n},\{\Hm(t)\}_{t=1}^{n-1},\{\hat{\Hm}(t)\}_{t=1}^{n},{\hv_{j1}}(n), {\hv_{j2}}(n)\right) \label{eq:dec-fun}
    \end{align}
    for the Rx-1 when $j=1$, and it is similarly defined for the Rx-2 by replacing $y(t)$ with $z(t)$,
\end{itemize}
such that the average decoding error probability $P_{e}^{(n)}$, defined as
\begin{align}
  P_{e}^{(n)} \defeq \E [\P[(\Mc_1, \Mc_2) \neq (\hat{\Mc}_1,\hat{\Mc}_2)]] ,
\end{align}
vanishes as the code length $n$ tends to infinity. The capacity region $\Cc$ is defined as the set of all achievable rate pairs. Accordingly, the DoF region can be defined as follows:
\begin{definition} [the degree-of-freedom region]
  The degree-of-freedom (DoF) region for two-user MISO interference channel is defined as
  \begin{align}
    \Dc &= \left\{ (d_1,d_2)\in \mathbb{R}_{+}^2 | \forall (w_1,w_2) \in \mathbb{R}_{+}^2, w_1d_1+w_2d_2 \le \limsup_{P \to \infty} \left( \sup_{(R_1,R_2) \in \Cc} \frac{w_1R_1+w_2R_2}{\log P}\right) \right\}.
  \end{align}
\end{definition}
Consequently, the DoF region for the two-user time-correlated MISO interference channel is stated in the following theorem.
\begin{theorem}
  In the two-user MISO interference channel with perfect delayed CSIT and imperfect current CSIT (as stated in Assumption~2), the optimal DoF region can be characterized by
  \begin{subequations}
  \begin{align}
    d_1 &\le 1 \label{eq:dof-bound-1}\\
    d_2 &\le 1 \label{eq:dof-bound-2}\\
    2 d_1 + d_2 &\le 2+\alpha \label{eq:dof-bound-3}\\
    d_1 + 2 d_2 &\le 2+\alpha \label{eq:dof-bound-4}.
  \end{align}
  \end{subequations}
\end{theorem}

\textbf{Remark}: Interestingly, the above DoF region is identical to that of the two-user MISO broadcast channel with perfect delayed CSIT and imperfect current CSIT \cite{Yang:2012,Gou:2012}. In fact, this result is consistent with previous results on the pure delayed CSIT case ($\alpha=0$) where it was shown that the DoF region for the two-user BC and the two-user IC coincides, and also on the special case of perfect instantaneous CSIT ($\alpha=1$).

For illustration, the DoF region for the two-user MISO IC is provided in Fig.~2. The DoF regions with no CSIT, pure perfect delayed CSIT, and perfect instantaneous CSIT are also plotted for comparison. It shows that the DoF region with perfect delayed CSIT and imperfect current CSIT is strictly larger than that with pure delayed CSIT and quickly approaches the region with perfect CSIT as the quality of current CSIT increases.

\begin{figure}[htb]
\begin{center}
\includegraphics[width=0.55\columnwidth]{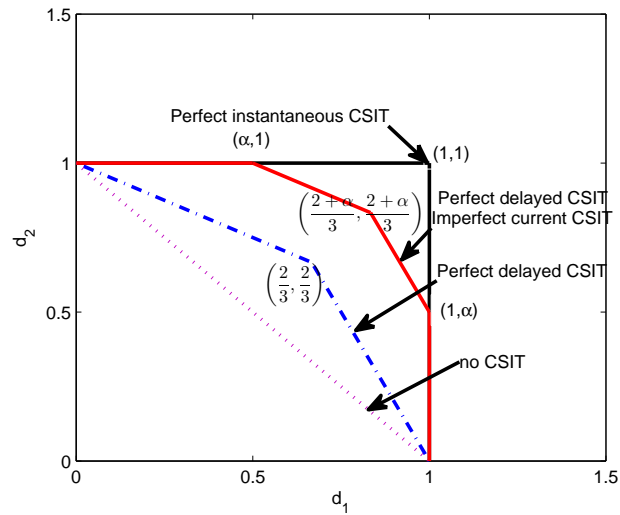}
\caption{DoF region for the two-user MISO interference channel (when $\alpha=0.5$).}
\label{fig:DoF}
\end{center}
\end{figure}

Given an $\alpha$, the DoF region is a polygon whose vertices are: $(0,1)$, $(\alpha,1)$, $\left(\frac{2+\alpha}{3},\frac{2+\alpha}{3}\right)$, $(1,\alpha)$ and $(1,0)$. In the following, we first characterize the outer bound, and then propose two schemes to show they are achievable, and in turn the entire region can be achieved by time sharing.

\section{Outer Bound}
We adopt a strategy reminisced in \cite{Vaze:2011} to obtain the genie-aided outer bound, by assuming that (i) both receivers know the CSI ${\Hm}(t)$ perfectly and instantaneously as well as the imperfect current CSI $\hat{\Hm}(t)$ at time $t$, and (ii) the Rx-2 has the instantaneous knowledge of the Tx-1's received signal $y(t)$.

Define
\begin{align}
y'(t) &\defeq \hv_{12}^\H(t) \xv_2(t) + e(t)\\
  z'(t) &\defeq \hv_{22}^\H(t) \xv_2(t) + b(t)\\
  \Tc &\defeq \{\Hm(t),\hat{\Hm}(t)\}_{t=1}^{n}\\
  \Uc(t) &\defeq \left\{\{ y'(i) \}_{i=1}^{t-1},\{ z'(i) \}_{i=1}^{t-1},\{\Hm(i)\}_{i=1}^{t-1} ,\{\hat{\Hm}(i)\}_{i=1}^{t} \right\}
\end{align}
where $\Tc$ denotes the channel state information and its estimated version available at the receivers from the beginning up to time instant $n$.

To ease our presentation, we denote:
\begin{align}
  n\epsilon_n \defeq 1+nR P_{e}^{(n)}
\end{align}
where $\epsilon_n$ tends to zero as $n \to \infty$ by the assumption that $\lim_{n \to \infty}P_{e}^{(n)}=0$. Then, we can upper-bound the achievable rate of Rx-1 by applying Fano's inequality:
{\small
\begin{align}
&nR_1 \\
&\le I(\Mc_1;\{y(t)\}_{t=1}^{n}|\Tc) + n \epsilon_n\\
&= I(\Mc_1,\Mc_2;\{y(t)\}_{t=1}^{n}|\Tc) - I(\Mc_2;\{y(t)\}_{t=1}^{n}|\Mc_1,\Tc) + n \epsilon_n  \\
&\le n \log P - I(\Mc_2;\{y(t)\}_{t=1}^{n}|\Mc_1,\Tc) +n \cdot O(1) + n \epsilon_n \label{eq:R-1-bound}\\
&= n \log P - h(\{y(t)\}_{t=1}^{n}|\Mc_1,\Tc) + h(\{y(t)\}_{t=1}^{n}|\Mc_1,\Mc_2,\Tc) +n \cdot O(1) + n \epsilon_n\\
&= n \log P - h(\{y(t)\}_{t=1}^{n}|\Mc_1,\Tc) +n \cdot O(1) + n \epsilon_n \label{eq:epsilon}\\
&= n \log P - h(\{ y'(t) \}_{t=1}^{n}|\Tc) +n \cdot O(1) + n \epsilon_n  \label{eq:R-1-remove-m1}\\
&\le n \log P - \sum_{t=1}^n h( y'(t)|\Tc,\{ y'(i) \}_{i=1}^{t-1},\{ z'(i) \}_{i=1}^{t-1}) +n \cdot O(1) + n \epsilon_n \label{eq:R-1-remove-m2}\\
&= n \log P - \sum_{t=1}^n h( y'(t)|\Uc(t),\Hm(t)) +n \cdot O(1) + n \epsilon_n \label{eq:R1-bound}
\end{align}
}
where (\ref{eq:R-1-bound}) follows the fact that the rate of the point-to-point MISO channel (i.e., Tx-1 together with Tx-2 are treated as the transmitter by cooperation while Rx-1 as the receiver) is bounded by $n\log P + n \cdot O(1)$; (\ref{eq:epsilon}) is due to the fact that (a) transmitted signals $\{\xv_i(t)\}_{t=1}^n$ are determined given messages, channel matrices up to $n$ and the encoding functions defined in~(\ref{eq:enc-fun}), (b) translation does change differential entropy, and (c) noise is independent of the channel matrices, the transmitted signals and the messages; (\ref{eq:R-1-remove-m1}) is obtained because (a) the transmitted signals $\{\xv_1(t)\}_{t=1}^n$ are determined provided the channel matrices, $\Mc_1$ and the encoding functions according to~(\ref{eq:enc-fun}), and (b) translation preserves differential entropy; (\ref{eq:R-1-remove-m2}) follows the chain rule of differential entropy and the fact that conditioning reduces differential entropy; the last equality is obtained due to $y'(t)$ is independent of $\{\Hm(k)\}_{k=t+1}^n$ and $\{\hat{\Hm}(k)\}_{k=t+1}^n$.

By applying Fano's inequality, we then also upper-bound the achievable rate of Rx-2 as
\begin{align}
&nR_2 \\
&\le I(\Mc_2;\{y(t)\}_{t=1}^{n},\{z(t)\}_{t=1}^{n}|\Tc) + n \epsilon_n\\
&\le I(\Mc_2;\{y(t)\}_{t=1}^{n},\{z(t)\}_{t=1}^{n},\Mc_1|\Tc) + n \epsilon_n \\
&= I(\Mc_2;\{y(t)\}_{t=1}^{n},\{z(t)\}_{t=1}^{n}|\Mc_1,\Tc) + n \epsilon_n \label{eq:R-2-chain-rule}\\
&= I(\Mc_2;\{y'(t)\}_{t=1}^{n},\{z'(t)\}_{t=1}^{n}|\Tc) + n \epsilon_n \label{eq:R2-remove-M1}\\
&= \sum_{t=1}^n I(\Mc_2;y'(t),z'(t)|\Tc,\{ y'(i) \}_{i=1}^{t-1},\{ z'(i) \}_{i=1}^{t-1}) + n \epsilon_n \nonumber \\
&\le \sum_{t=1}^n  I(\xv_2(t);y'(t),z'(t)|\Tc,\{ y'(i) \}_{i=1}^{t-1},\{ z'(i) \}_{i=1}^{t-1}) + n \epsilon_n \label{eq:R2-markov-chain}\\
&= \sum_{t=1}^n \left(h(y'(t),z'(t)|\Tc,\{ y'(i) \}_{i=1}^{t-1},\{ z'(i) \}_{i=1}^{t-1}) \right. \nonumber \\
&~~~~\left. -h(y'(t),z'(t)|\xv_2(t),\Tc,\{ y'(i) \}_{i=1}^{t-1},\{ z'(i) \}_{i=1}^{t-1} ) \right) + n \epsilon_n \label{eq:R2-noise-independence}\\
&\le \sum_{t=1}^n h(y'(t),z'(t)|\Tc,\{ y'(i) \}_{i=1}^{t-1},\{ z'(i) \}_{i=1}^{t-1})  + n \epsilon_n \label{eq:last-eq-2}\\
&=\sum_{t=1}^n h(y'(t),z'(t)|\Uc(t),\Hm(t)) + n \epsilon_n \label{eq:R2-bound}
\end{align}
where (\ref{eq:R-2-chain-rule}) is obtained because of the chain rule of mutual information and the independence between $\Mc_1$ and $\Mc_2$; (\ref{eq:R2-remove-M1}) is due to (a) the transmitted signals $\{\xv_1(t)\}_{t=1}^n$ are determined given message $\Mc_1$, channel matrices and encoding functions, and
(b) $\Mc_2$ and $\{y'(t),z'(t)\}$ are independent of $\Mc_1$; (\ref{eq:R2-markov-chain}) is obtained by Markov chain $\Mc_2 \to \xv_2(t) \to \{y'(t),z'(t)\}$ and the data processing inequality; (\ref{eq:last-eq-2}) is because (a) translation does not change differential entropy, (b) Gaussian noise terms are independent from instant to instant, and are also independent of the channel matrices and the transmitted signals, and (c) the differential entropy of Gaussian noise is nonnegative; and the last equality is obtained due to the independence $\{y'(t),z'(t)\}$ of $\{\Hm(k)\}_{k=t+1}^n$ and $\{\hat{\Hm}(k)\}_{k=t+1}^n$.

According to the Markov chain $\{\xv_2(t)\}_{t=1}^n \to \left(\{y(t)\}_{t=1}^n, \{z(t)\}_{t=1}^n\right) \to \{y(t)\}_{t=1}^n$, we upper-bound the weighted sum rate as
\begin{align}
  &n(2R_1+R_2) \\
  &\le 2n \log P + \sum_{t=1}^n \left(h(y'(t),z'(t)|\Uc(t),\Hm(t)) - 2 h(y'(t)|\Uc(t),\Hm(t))\right) + n \cdot O(1) + n \epsilon_n \label{eq:R1-2R2}.
\end{align}

Before preceding further, we introduce the following lemma stated in~\cite{Yang:2012}.
\begin{lemma} [\cite{Yang:2012}]
  For an $m \times 1$ random vector $\hv = \hat{\hv}+\tilde{\hv}$ where $\tilde{\hv} \sim \CN(0,\sigma^2 \Id_m)$ is independent of $\hat{\hv}$, given any $\Km \succeq 0$ with eigenvalues $\lambda_1 \ge \cdots \ge \lambda_m$, we have the following upper and lower bounds:
  \begin{align}
    \log (1+e^{-\gamma} \sigma^2 \lambda_1) + O(1) &\le \E_{\tilde{\hv}} \log (1+\hv^\H \Km \hv) \le \log (1+\norm{\hat{\hv}}^2 \lambda_1) + O(1).
  \end{align}
  The difference of the upper and lower bounds can be further bounded by
  \begin{align}
    &\log (1+\norm{\hat{\hv}}^2 \lambda_1) - \log (1+e^{-\gamma} \sigma^2 \lambda_1) \le - \log (\sigma^2) + O(1)
  \end{align}
  where $\gamma$ is Euler's constant.
\end{lemma}

With the definitions
\begin{align}
  \Sm(t) &\eqdef \Bmatrix{\hv_{12}^\H(t) \\ \hv_{22}^\H(t) }\\
  \hat{\Sm}(t) &\eqdef \Bmatrix{\hat{\hv}_{12}^\H(t) \\ \hat{\hv}_{22}^\H(t) }\\
  \wv(t) &\eqdef [e(t) \  b(t) ]^\T\\
  \Km(t) &\eqdef \E\{\xv_2(t)\xv_2^{\H}(t)|\Uc(t)\}
\end{align}
we further upper-bound the weighted difference of two conditional differential entropies that derived above, i.e.,
{\small
\begin{align}
  &h(y'(t),z'(t)|\Uc(t),\Hm(t)) - 2 h(y'(t)|\Uc(t),\Hm(t))\\
  &=h(y'(t),z'(t)|\Uc(t),\Sm(t)) - 2 h(y'(t)|\Uc(t),\Sm(t)) \label{eq:gaussian-input-1}\\
  &\le \max_{\substack{\Cm \succeq 0,\\ \trace(\Cm) \le P}} \max_{\substack{p(\Uc(t)),\\p(\xv_2(t)|\Uc(t)) \\ \Km(t) \preceq \Cm}} (h(y'(t),z'(t)|\Uc(t),\Sm(t)) - 2 h(y'(t)|\Uc(t),\Sm(t)))\\
  &= \max_{\substack{\Cm \succeq 0,\\ \trace(\Cm) \le P}} \max_{\substack{p(\Uc(t))\\ \Km_u(t) \preceq \Cm}} ( h( \Sm(t) \uv(t)+\wv(t) |\Uc(t),\Sm(t)) - 2 h( \hv_{12}^\T(t) \uv(t)+e(t) |\Uc(t),\Sm(t)) ) \label{eq:gaussian-input-2}\\
  &=  \max_{\substack{\Cm \succeq 0,\\ \trace(\Cm) \le P}}\max_{\substack{p(\hat{\Sm}(t))\\ \Km_u(t) \preceq \Cm}} ( h( \Sm(t) \uv(t)+\wv(t) |\Sm(t),\hat{\Sm}(t)) - 2 h( \hv_{12}^\H(t) \uv(t)+e(t) |\Sm(t),\hat{\Sm}(t)) ) \label{eq:gaussian-input-3}\\
  &=  \max_{\substack{\Cm \succeq 0,\\ \trace(\Cm) \le P}} \max_{\substack{p(\hat{\Sm}(t))\\ \Km_u(t) \preceq \Cm}} \E_{\Sm(t),\hat{\Sm}(t)} (\log \det (\mathbf I + \Sm(t) \Km_u(t) \Sm^{\H}(t) ) - 2 \log(1+\hv_{12}^\H(t) \Km_u(t) \hv_{12}(t))) \label{eq:expectation-bound1}\\
   &\le \E_{\hat{\Sm}(t)}  \max_{\substack{\Cm \succeq 0,\\ \trace(\Cm) \le P}}\max_{\substack{p(\hat{\Sm}(t))\\ \Km_u(t) \preceq \Cm}} \E_{\Sm(t)|\hat{\Sm}(t)} (\log \det (\mathbf I + \Sm(t) \Km_u(t) \Sm^{\H}(t) ) - 2 \log(1+\hv_{12}^\H(t) \Km_u(t) \hv_{12}(t))) \label{eq:expectation-bound3}\\
   &\le \E_{\hat{\Sm}(t)}  \max_{\substack{\Cm \succeq 0,\\ \trace(\Cm) \le P}}\max_{\substack{p(\hat{\Sm}(t))\\ \Km_u(t) \preceq \Cm}} \E_{\Sm(t)|\hat{\Sm}(t)} (\log (1+\hv_{22}^\H(t) \Km_u(t) \hv_{22}(t))) - \log(1+\hv_{12}^\H(t) \Km_u(t) \hv_{12}(t))) \label{eq:expectation-bound4}\\
   &\le \E_{\hat{\Sm}(t)}  \max_{\substack{\Cm \succeq 0,\\ \trace(\Cm) \le P}}\max_{\substack{p(\hat{\Sm}(t))\\ \Km_u(t) \preceq \Cm}} \left(\log (1+\Norm{\hat{\hv}_{22}(t)} \lambda_1(\Km_u(t))) - \log \left(1+ e^{-\gamma} \sigma^2 \lambda_1(\Km_u(t))\right)\right) + O(1)\\
   & \le \alpha \log P + O(1) \label{eq:diff-bound}
\end{align}
}
where in (\ref{eq:gaussian-input-1}) $\Hm(t)$ is replaced by $\Sm(t)$ because of the independence of $\{y'(t),z'(t)\}$; (\ref{eq:gaussian-input-2}) is obtained because Gaussian distributed vector $\uv(t)$ maximizes the weighted difference of two differential entropies over all conditional distribution of $\xv_2(t)$ with the same covariance matrix constraint, where $\Km_u(t) \defeq \E\{\uv(t)\uv^{\H}(t)\} =  \max_{p(\Uc(t))} \Km(t) $~\cite{Liu:2007}; (\ref{eq:gaussian-input-3}) is because $\uv(t)$, $\Sm(t)$ and $\wv(t)$ are independent of $\Uc(t)$ except $\hat{\Sm}(t)$; (\ref{eq:expectation-bound1}) is obtained because $\uv(t)$ is Gaussian distributed and independent of $\{\Hm(t)\}_{t=1}^{n}$, $\{\hat{\Hm}(t)\}_{t=1}^{n}$ as well as the noise terms; (\ref{eq:expectation-bound3}) follows the fact that putting the expectation out of the maximization increases the value; (\ref{eq:expectation-bound4}) follows from the inequality $\det(\Id+\Am) \le \prod_{i=1}^m \left(1+a_{ii}\right)$ where $\Am$ is an $m \times m$ positive semidefinite matrix with entry $a_{ij}$; the last two inequalities are according to Lemma~1 and the quality of current CSIT ($\sigma^2 \sim P^{-\alpha}$).

Accordingly, we obtain an upper bound of $2R_1+R_2$ from (\ref{eq:R1-2R2}) and (\ref{eq:diff-bound}), i.e.,
\begin{align}
  &n(2R_1+R_2) \le n (2+\alpha) \log P + n \cdot O(1) + n \epsilon_n
\end{align}
as $n\to \infty$, from which (\ref{eq:dof-bound-3}) is obtained according to the definition of DoF.

By exchanging the roles of Tx-1/Rx-1 and Tx-2/Rx-2, another inequality (\ref{eq:dof-bound-4}) can be similarly obtained by assuming Rx-1 has the instantaneous knowledge of $z(t)$, where the weighted rate is bounded as
\begin{align}
  &n(R_1+2R_2) \le n (2+\alpha) \log P + n \cdot O(1) + n \epsilon_n.
\end{align}

Together with the first two bounds~(\ref{eq:dof-bound-1}) and (\ref{eq:dof-bound-2}) which are obtained by the constraint of antenna configuration, the DoF region is completely characterized.

\section{Achievability}

With perfect delayed CSIT, the authors in~\cite{Vaze:2011} and~\cite{Ghasemi:2011} characterize the DoF region for two-user MIMO interference channel, bridging between the case with no CSIT~\cite{Huang:2012,Zhu:2012,Vaze:2009} and that with perfect instantaneous CSIT~\cite{Jafar:2007}. Particularly, for two-user MISO case, the DoF pair $(\frac{2}{3},\frac{2}{3})$ is achievable with delayed CSIT, strictly larger than $(\frac{1}{2},\frac{1}{2})$ achieved with no CIST and dominated by $(1,1)$ with perfect CSIT.

The technique exploits the advantage of interference alignment in the time domain by utilizing the delayed CSIT together with the space domain, which is referred to as MAT alignment~\cite{Maddah-Ali:10,Ghasemi:2011,Vaze:2011}. We first briefly review its application in the interference channel.
\subsection{MAT in the Interference Channel}
The MAT alignment in the interference channel is an extension from the broadcast channel, taking into account the distributive and uncooperative nature of the transmitters~\cite{Ghasemi:2011}. The two-phase protocol which consumes three time slots is described as follows:
\subsubsection*{Phase-I} In this phase, each Tx transmits two independent encoded symbols to its intended receiver without precoding during a single time slot, i.e.,
  \begin{subequations}
  \begin{align}
    \xv_1(1) &= \uv\\
    \xv_2(1) &= \vv
  \end{align}
  \end{subequations}
  and the received signals at both receivers are
  \begin{subequations}
  \begin{align}
    y(1) &= \hv_{11}^\H(1) \uv + \underbrace{\hv_{12}^\H(1) \vv}_{\eta_1} + e(1) \\
    z(1) &= \underbrace{\hv_{21}^\H(1) \uv}_{\eta_2} + \hv_{22}^\H(1) \vv + b(1),
  \end{align}
  \end{subequations}
  where $\eta_1$ and $\eta_2$ are interference terms overheard at the Tx-1 and Tx-2, respectively.
\subsubsection*{Phase-II} At the end of phase-I, the delayed CSIT $\{\hv_{21}(1)\}$ is available at the Tx-1, while $\{\hv_{12}(1)\}$ is accessible at the Tx-2. Together with the transmitted symbols, the overheard interference terms are reconstructible at both Txs. By retransmitting the overheard interference terms $\eta_2=\hv_{21}^\H(1)\uv$ at the Tx-1 and $\eta_1=\hv_{12}^\H(2)\vv$ at the Tx-2 with time division, i.e.,
      \begin{subequations}
      \begin{align}
        \xv_1(2) &= \Bmatrix{\eta_2 \\ 0} \\
        \xv_2(2) &= \mathbf 0
      \end{align}
      \end{subequations}
      and
      \begin{subequations}
      \begin{align}
        \xv_1(3)&=\mathbf 0 \\
        \xv_2(3)&=\Bmatrix{\eta_1\\0}
      \end{align}
      \end{subequations}
      where two entire time slots are consumed, we cancel the interference terms $\eta_1$ and $\eta_2$ at the Rx-1 and Rx-2, and importantly provide another linear combination of $\uv$ (from $\eta_2$) and $\vv$ (from $\eta_1$) to the Rx-1 and Rx-2, respectively. By the end of phase-II, both receivers are able to recover their own symbols with high probability.
The key idea behind is interference repetition and alignment in both space and time domain. At each receiver, the mutual interference aligns in one dimension, while the desired signal spans in a two-dimensional space. This enables each receiver to retrieve the desired signal from a three-dimensional space.

\subsection{Integrating the Imperfect Current CSIT}
The MAT alignment takes into account the completely outdated CSIT, regardless of the correlation between current and previous channel states. As a matter of fact, such an assumption on the delayed CSIT is over-pessimistic, since the current CSI can be predicted from the past states if the underlying channel exhibits some temporal correlation. Recent results demonstrate that the DoF region can be enlarged in broadcast channel by using estimated current CSIT, even it is imperfect~\cite{Kobayashi:2012,Yang:2012}.

In the following, two schemes are proposed, demonstrating the larger DoF region can be achieved by utilizing estimated current CSIT exploited from time correlation model. Instead of forwarding the interference terms in an analog fashion~\cite{Ghasemi:2011,Vaze:2011}, we first quantize the interference and then retransmit the quantized version. By utilizing the imperfect current CSIT for precoding, the interference terms are efficiently compressed with quantization.

In the following two schemes, we demonstrate the vertices $\left(\frac{2+\alpha}{3},\frac{2+\alpha}{3}\right)$ and $(1,\alpha)$ are all achievable. Note that we simply use $\hat{\hv}_{ji}(t)$ for the range space while $\hat{\hv}_{ji}^{\perp}(t)$ for the null space of $\hat{\hv}_{ji}(t)$. The precoder design to improve the achievable rate is out of scope of this paper.

\subsubsection{\textbf{Achievability of $\left(\frac{2+\alpha}{3},\frac{2+\alpha}{3}\right)$}}
Inspired by the enhanced scheme for the two-user MISO broadcast channel~\cite{Yang:2012}, a 3-time-slotted protocol, which achieves the vertex $\left(\frac{2+\alpha}{3},\frac{2+\alpha}{3}\right)$ of the DoF region for the two-user MISO interference channel, is detailed as follows.
\subsubsection*{Slot-1}
   In the first time slot, the symbol vectors $\uv(1)$ and $\vv(1)$ are respectively sent from the two transmitters with precoding, heading to their corresponding receivers:
  \begin{subequations}
  \begin{align}
    \xv_1(1) &= [\hat{\hv}_{21}(1) \ \hat{\hv}_{21}^{\perp}(1)] \uv(1)\\
    \xv_2(1) &= [\hat{\hv}_{12}(1) \ \hat{\hv}_{12}^{\perp}(1)] \vv(1)
  \end{align}
  \end{subequations}
  where $\uv(1)=[u_1(1) \ u_2(1)]^{\T}$, $\vv(1)=[v_1(1) \ v_2(1)]^{\T}$ satisfy $\E(\Norm{\uv(1)})=\E(\Norm{\vv(1)}) \le P$. The received signal at both receivers are then given as:
  \begin{subequations}
    \begin{align}
    y(1) &= \hv_{11}^\H(1) \xv_1(1) + {\eta}_1 + e(1) \\
    z(1) &= \hv_{22}^\H(1) \xv_2(1) + {\eta}_2 + b(1),
    \end{align}
  \end{subequations}
  where ${\eta}_1$ and ${\eta}_2$ are interferences overheard at the Rx-1 and Rx-2 respectively, i.e.,
  \begin{subequations}
    \begin{align}
    {\eta}_1 &= {\hv}_{12}^{\H}(1) \hat{\hv}_{12}(1) v_1(1)+{\hv}_{12}^{\H}(1) \hat{\hv}_{12}^{\perp}(1) v_2(1)\\
    {\eta}_2 &= {\hv}_{21}^{\H}(1) \hat{\hv}_{21}(1) u_1(1)+{\hv}_{21}^{\H}(1) \hat{\hv}_{21}^{\perp}(1) u_2(1).
    \end{align}
  \end{subequations}
  According to (\ref{eq:P-to-alpha}), i.e., $\E (\Abs{{\hv}_{ji}^{\H}(1) \hat{\hv}_{ji}^{\perp}(1)}) \sim P^{-\alpha}$, we can make $\E(|{\eta}_1|^2)=\E(|{\eta}_2|^2) \sim P^{1-\alpha}$ by allocating $\E(|u_1(1)|^2)=\E(|v_1(1)|^2) = P^{1-\alpha}$ whereas $\E(|u_2(1)|^2)=\E(|v_2(1)|^2) = P-P^{1-\alpha} \sim P$.


  At the end of slot-1, Tx-1 can reconstruct $\eta_2 ={\hv}_{21}^\H(1)\xv_2(1)$ while Tx-2 can reconstruct $\eta_1 = {\hv}_{21}^{\H}(1) \xv_1(1)$. Instead of forwarding the interferences in an analog fashion, we first quantize the interference term $\eta_i$ into $\hat{\eta}_i$ with $(1-\alpha) \log P$ bits each, then encode the index of $\hat{\eta}_i$ to codeword $c_i$ using a Gaussian channel codebook, and forward $c_i$ as a common message to both receivers in the ensuing two time slots. To ease our presentation, the process of the encoding and decoding of $c_i$ is omitted hereafter, making it look as if the codeword $\hat{\eta}_i$ itself is conveyed to the receivers\footnote{The simplification does not affect our results as long as we consider DoF region. For the rate region, the general scheme and more rigorous proof can be straightforwardly extended from~\cite{Yang:2012}.}.

  The source codebook $\Xc_1$ (resp.~$\Xc_2$) is generated for ${\eta}_2$ (resp.~ ${\eta}_1$) and maintained at the Tx-1 (resp.~Tx-2). The entry $\hat{\eta}_i$ in codebook $\Xc_i$ satisfies
  \begin{align}
    {\eta}_i=\hat{\eta}_i+\Delta_i
  \end{align}
  where $\Delta_i$ is the quantization error with distortion $\E(|\Delta_i|^2) \sim \sigma^2_{{\eta}_i} D$ and independent of $\hat{\eta}_i$. According to the rate distortion theory\cite{Cover:2006}, we let the normalized distortion $D$ decay as $P^{-(1-\alpha)}$ (in turn $\E(|\Delta_i|^2) \sim P^{0}$) so that each receiver can decode it successfully and the quantization error is drowned in the noise.

  \subsubsection*{Slot-2}  During the second time slot, the index corresponding to $\hat{\eta}_2$ is encoded to $c_2$ and sent from Tx-1 as a common message together with a new symbol $u(2)$ with ZF precoding, while a new symbol $v(2)$ intended to Rx-2 is instantaneously sent from Tx-2 with ZF precoding as well. By omitting the encoding and decoding process of $c_2$, the equivalent transmitted signals can be written as
      \begin{subequations}
      \begin{align}
        \xv_1(2) &= \Bmatrix{P^{\alpha/2} \hat{\eta}_2 \\ 0} + \hat{\hv}_{21}^{\perp}(2) u(2)\\
        \xv_2(2) &= \hat{\hv}_{12}^{\perp}(2) v(2).
      \end{align}
      \end{subequations}
      where the cordword $\hat{\eta}_2$ is power scaled with $P^{\alpha/2}$ to ensure it can be recovered from noisy observation.
      To avoid interference from the other transmitters, we assume the new symbols $u(2)$ and $v(2)$ satisfy the power constraint $\E (|u(2)|^2)=\E (|v(2)|^2) \le P^{\alpha}$. The received signals at both receivers are given as:
      \begin{subequations}
        \begin{align}
        y(2) &= \underbrace{h_{11,1}^{*}(2) {P^{\alpha/2} \hat{\eta}_2 }}_{t_{11}} + \underbrace{\hv_{11}^\H(2) \hat{\hv}_{21}^{\perp}(2) u(2)}_{t_{12}} + \underbrace{\hv_{12}^\H(2) \hat{\hv}_{12}^{\perp}(2) v(2)}_{t_{13}} + e(2) \\
        z(2) &= {h_{21,1}^{*}(2) {P^{\alpha/2} \hat{\eta}_2  }}+ {\hv_{22}^\H(2) \hat{\hv}_{12}^{\perp}(2) v(2)}+ {\hv_{21}^\H(2) \hat{\hv}_{21}^{\perp}(2) u(2)} + b(2).
        \end{align}
      \end{subequations}
      Note that in the received signal $y(2)$, $\E (|t_{11}|^2) \sim P$, $\E (|t_{12}|^2) \sim P^{\alpha}$, while $\E (|t_{13}|^2) \sim P^0$ is at noise level. With distortion $D \sim P^{-(1-\alpha)}$, both receivers can retrieve $\hat{\eta}_2$ with high probability by treating $t_{21}$ and $t_{22}$ as noise~\cite{Yang:2012}. By removing $\hat{\eta}_2$ from the received signals, $u(2)$ and $v(2)$ can be recovered with high probability as long as their power constraints are satisfied.
  \subsubsection*{Slot-3} The transmission in the third time slot is similar to that in the slot-2, where the index corresponding to $\hat{\eta}_1$ chosen from $\Xc_2$ is encoded to $c_1$ and transmitted as a common message together with another new symbol $v(3)$ from Tx-2, while only one new symbol $u(3)$ intended to Rx-1 is sent from Tx-1. By omitting the encoding and decoding process of $c_1$, the equivalent transmitted signals are
      \begin{subequations}
      \begin{align}
        \xv_1(3) &= \hat{\hv}_{21}^{\perp}(3) u(3)\\
        \xv_2(3) &= \Bmatrix{P^{\alpha/2} \hat{\eta}_1 \\ 0} + \hat{\hv}_{12}^{\perp}(3) v(3)
      \end{align}
      \end{subequations}
      where the new symbols $u(3)$ and $v(3)$  satisfy the power constraint $\E (|u(3)|^2)=\E (|v(3)|^2) \le P^{\alpha}$. The received signals at both receivers are given as
      \begin{subequations}
        \begin{align}
        y(3) &= {h_{12,1}^{*}(3) {P^{\alpha/2} \hat{\eta}_1 }} + {\hv_{11}^\H(3) \hat{\hv}_{21}^{\perp}(3) u(3)} + {\hv_{12}^\H(3) \hat{\hv}_{12}^{\perp}(3) v(3)} + e(3) \\
        z(3) &= {h_{22,1}^{*}(3) {P^{\alpha/2} \hat{\eta}_1 }} + {\hv_{22}^\H(3) \hat{\hv}_{12}^{\perp}(3) v(3)} + {\hv_{21}^\H(3) \hat{\hv}_{21}^{\perp}(3) u(3)} + b(3).
        \end{align}
      \end{subequations}
      Similarly to the slot-2, $\hat{\eta}_1$ is retrievable at both receivers by treating other terms as noise, and $u(3)$ and $v(3)$ can be also recovered respectively by subtracting $\hat{\eta}_1$ from the received signals at both receivers.

At the end of the third slot, $u(2)$, $u(3)$, $\hat{\eta}_1$ and $\hat{\eta}_2$ can be successfully recovered at the Rx-1. As was modeled in~\cite{Kobayashi:2012,Yang:2012}, an equivalent MIMO can be formulated to find the symbols $\uv(1)$:
\begin{align}
  \Bmatrix{y(1)-\hat{\eta}_1 \\ \hat{\eta}_2} = \Bmatrix{\hv_{11}^\H(1) \\ \hv_{21}^\H(1)} \xv_1(1) + \Bmatrix{e(1)+\Delta_1 \\ -\Delta_2}
\end{align}
for the Rx-1, and it is similar for the Rx-2.

\begin{lemma}
  The vertex $\left(\frac{2+\alpha}{3},\frac{2+\alpha}{3}\right)$ of DoF region is achievable by the above scheme.
\end{lemma}
\begin{proof}
    We outline the main idea of the proof here and please refer to Appendix for details.

    At the Rx-1 for instance, we transform the original signal model into an equivalent $2 \times 2$ point-to-point MIMO system model for $\uv(1)$ (resp.~$\vv(1)$ at the Rx-2), together with two parallel SISO signal models respectively for $u(2)$ and $u(3)$ (resp.~$v(2)$ and $v(3)$ at the Rx-2). For the MIMO model, we obtain the DoF of $2-\alpha$, while get $\alpha$ DoF for each parallel SISO model, and finally $\frac{2-\alpha+2 \alpha}{3} = \frac{2+\alpha}{3}$ DoF is achieved per user.
\end{proof}

\subsubsection{\textbf{Achievability of $(1,\alpha)$}}
In the following, we extend the Han-Kobayashi scheme~\cite{Han:81} here to achieve the vertex $(1,\alpha)$.

The symbols sent from the Tx-1 consists of two parts $u_c$ and $u_p$, where only $u_p$ is precoded by using imperfect current CSIT. Simultaneously, one symbol $v_p$ intended to Rx-2 is sent from the Tx-2 with ZF precoding. The transmission can be given as
  \begin{subequations}
  \begin{align}
    \xv_1 &= \Bmatrix{u_c\\0} + \hat{\hv}_{21}^{\perp} u_p \\
    \xv_2 &= \hat{\hv}_{12}^{\perp} v_p
  \end{align}
  \end{subequations}
  where the transmitted symbols are assumed to satisfy the power constraints $\E(\Abs{u_c}) \le P$, $\E(\Abs{u_p}) = \E(\Abs{v_p}) \le P^{\alpha}$. Although the symbol $u_c$ is decodable by both receivers and hence referred to as a common message, it is only desirable by Rx-1. On the other hand, we refer to $u_p$, $v_p$ as the private messages which can only be seen and decoded by their corresponding receivers.

At the receiver side, we have
  \begin{subequations}
    \begin{align}
    y &= h_{11,1}^{*} u_c + \underbrace{\hv_{11}^\H \hat{\hv}_{21}^{\perp} u_p}_{\eta_{11}} + \underbrace{{\hv}_{12}^{\H} \hat{\hv}_{12}^{\perp} v_p}_{{\eta}_{12}} + e \\
    z &= h_{21,1}^{*} u_c + \underbrace{\hv_{22}^\H \hat{\hv}_{12}^{\perp} v_p}_{\eta_{22}} + \underbrace{{\hv}_{21}^{\H}\hat{\hv}_{21}^{\perp} u_p}_{{\eta}_{21}} + b,
    \end{align}
  \end{subequations}
  where the terms carrying common message are with approximated power $P$, while those carrying private messages are $\E(|{\eta}_{11}|^2)=\E(|{\eta}_{22}|^2) \sim P^{\alpha}$, and the interference terms $\E(|{\eta}_{12}|^2)=\E(|{\eta}_{21}|^2) \sim P^{0}$ are at noise level according to (\ref{eq:P-to-alpha}).

By firstly treating ${\eta}_{i1}, {\eta}_{i2}$ as noise, Rx-$i$ can recover the common message $u_c$ with high probability. Then, the private messages $u_p$ and $v_p$ can be retrieved from the received signals after $u_c$ being subtracted at the Rx-1 and Rx-2, respectively.

\begin{lemma}
  The vertices $(1,\alpha)$ and $(\alpha,1)$ of DoF region are achievable.
\end{lemma}
\begin{proof}
Define
    \begin{subequations}
    \begin{align}
    y' &\defeq \eta_{11} + \eta_{12} + e \\
    z' &\defeq \eta_{22} + \eta_{21} + b.
    \end{align}
  \end{subequations}

    For both receivers, the achievable rate can be given by
    \begin{align}
      I(u_c,u_p;y|\Tc) &= I(u_c;y|\Tc)+I(u_p;y|\Tc,u_c)\\
      &=I(u_c;y|\Tc)+I(u_p;y'|\Tc)\\
      &=  \E \log \left( 1+\frac{\Abs{h_{11,1}^{*} u_c}}{\Abs{\eta_{11}}+\Abs{{\eta}_{12}}+\Abs{e}} \right) + \E \log \left( 1+\frac{\Abs{\eta_{11}}}{\Abs{{\eta}_{12}}+\Abs{e}} \right)\\
      &= (1-\alpha) \log P + \alpha \log P + O(1)\\
      &= \log P + O(1)
    \end{align}
    for the Rx-1, and
    \begin{align}
      I(v_p;z|\Tc) &=  I(v_p;z|\Tc,u_c) + I(v_p;u_c|\Tc) - I(v_p;u_c|\Tc,z)\\
      &= I(v_p;z|\Tc,u_c) = I(v_p;z'|\Tc) \label{eq:scheme-2-pr}\\
      &= \E \log \left( 1+\frac{\Abs{\eta_{22}}}{\Abs{{\eta}_{21}}+\Abs{b}} \right)\\
      &= \alpha \log P + O(1)
    \end{align}
    for the Rx-2, where (\ref{eq:scheme-2-pr}) holds because $u_c$ and $v_p$ are independent.

    The DoF for both receivers can be simply obtained by definition. The other vertex $(\alpha,1)$ can be achieved by swapping the roles of Tx-1 and Tx-2. This completes the proof.
\end{proof}

Note that the vertices $(1,0)$ and $(0,1)$ are achievable by letting one pair communicate while keeping the other one silent. In conclusion, all vertices of the DoF region for two-user MISO interference channel are achievable, and in turn the entire region can be achieved by time sharing.

\section{Extension to MIMO Case}
Here, we extend the aforementioned MISO case to a class of MIMO settings with antenna configuration $(M,M,N,N)$, where $M$ antennas at each transmitter and $N$ antennas at each receiver, satisfying $M \ge 2N$. This includes a generalized MISO setting with more than 2 antennas at each transmitter. The discrete time baseband signal model is given by
\begin{subequations}
\begin{align}
\yv(t) &= \Hm_{11}(t) \xv_1(t) + \Hm_{12}(t) \xv_2(t) + \ev(t) \\
\zv(t) &= \Hm_{21}(t) \xv_1(t) + \Hm_{22}(t) \xv_2(t) + \bv(t),
\end{align}
\end{subequations}
for any time instant $t$, where $\Hm_{ji}(t) \in \CC^{N\times M}$ is the channel matrix from Tx-$i$ to Rx-$j$; $\ev(t), \bv(t) \sim \CN[0,\Id_N]$ are normalized AWGN vectors at the respective receivers; the coded input signal $\xv_i(t) \in \CC^{M\times 1}$ is subject to the power constraint $\E( \norm{\xv_i(t)}^2 ) \le P$, $\forall\,t$.

In analogy to the MISO case, we have the optimal DoF region of two-user time correlated $(M,M,N,N)$ MIMO interference channel.
\begin{theorem}
  In the two-user $(M,M,N,N)$ MIMO interference channel ($M \ge 2N$) with perfect delayed CSIT and imperfect current CSIT, the optimal DoF region can be characterized by
  \begin{subequations}
  \begin{align}
    d_1 &\le N \\
    d_2 &\le N \\
    d_1+2d_2 &\le N(2+\alpha)\\
    2d_1+d_2 &\le N(2+\alpha).
  \end{align}
  \end{subequations}
\end{theorem}

\textbf{Remark}: The DoF region is irrelevant to the number of transmit antennas as long as $M \ge 2N$. For the $M \times 1$ MISO case ($M \ge 2$), the DoF region is identical to that when $M=2$, coinciding with the region for two-user MISO broadcast channel.

Following the same strategy as the MISO case, we first provide the outer bound and then show the region confined by the outer bound is achievable.

\subsection{Outer Bound}
The outer bound can be simply extended from the MISO case. To avoid redundancy, we outline the main difference but omit the similar parts. By defining similarly $\yv'(t)$, $\zv'(t)$, $\Uc(t)$, i.e.,
\begin{align}
  \yv'(t) &\defeq \Hm_{12}(t) \xv_2(t) + \ev(t)\\
  \zv'(t) &\defeq \Hm_{22}(t) \xv_2(t) + \bv(t)\\
  \Uc(t) &\defeq \left\{\{ \yv'(i) \}_{i=1}^{t-1},\{ \zv'(i) \}_{i=1}^{t-1},\{\Hm(i)\}_{i=1}^{t-1} ,\{\hat{\Hm}(i)\}_{i=1}^{t} \right\},
\end{align}
we have
\begin{align}
  nR_1 &\le n N \log P - \sum_{t=1}^n h( \yv'(t)|\Uc(t),\Hm(t)) + n \cdot O(1) + n \epsilon_n \\
  nR_2 &\le \sum_{t=1}^n h(\yv'(t),\zv'(t)|\Uc(t),\Hm(t))  + n \epsilon_n.
\end{align}
Define
\begin{align}
  \Sm(t) \defeq \Bmatrix{\Hm_{12}(t) \\ \Hm_{22}(t)} \quad \quad \hat{\Sm}(t) \defeq \Bmatrix{\hat{\Hm}_{12}(t) \\ \hat{\Hm}_{22}(t)}
\end{align}
and we upper-bound the weighted difference of two conditional differential entropies by
{\small
\begin{align}
  &h(\yv'(t),\zv'(t)|\Uc(t),\Hm(t)) - 2h( \yv'(t)|\Uc(t),\Hm(t))\\
  & \le \E_{\hat{\Sm}(t)}  \max_{\substack{\Cm \succeq 0,\\ \trace(\Cm) \le P}}\max_{\substack{p(\hat{\Sm}(t))\\ \Km_u(t) \preceq \Cm}} \E_{\Sm(t)|\hat{\Sm}(t)} (\log \det (\Id + \Sm(t) \Km_u(t) \Sm^{\H}(t) ) - 2  \log \det(\Id+\Hm_{12}(t) \Km_u(t) \Hm_{12}^\H(t))) \\
  & \le \E_{\hat{\Sm}(t)}  \max_{\substack{\Cm \succeq 0,\\ \trace(\Cm) \le P}}\max_{\substack{p(\hat{\Sm}(t))\\ \Km_u(t) \preceq \Cm}} \E_{\Sm(t)|\hat{\Sm}(t)} (\log \det (\Id + \Hm_{22}(t) \Km_u(t) \Hm_{22}^{\H}(t) ) - \log \det (\Id+\Hm_{12}(t) \Km_u(t) \Hm_{12}^\H(t))) \\
  & \le N \alpha \log P + O(1)
\end{align}
}
where $\Km_u(t)$ possesses the same definition as that in the MISO case and the last inequality is obtained according to the following Lemma.
\begin{lemma}
  For an $N \times M$ random matrix ($M\ge N$) $\Hm = \hat{\Hm}+\tilde{\Hm}$ where $\tilde{\Hm}$ is independent of $\hat{\Hm}$ and whose entries satisfy $\tilde{h}_{ij} \sim \CN(0,\sigma^2)$, given any $\Km \succeq 0$ with eigenvalues $\lambda_1 \ge \cdots \ge \lambda_{M}$, we have the following upper and lower bounds:
  \begin{align}
     \E_{\tilde{\Hm}} \log \det (\Id+\Hm \Km \Hm^\H) &\le \sum_{i=1}^{N}\log (1+\norm{\hat{\Hm}}^2 \lambda_i) + O(1)\\
     \E_{\tilde{\Hm}} \log \det (\Id+\Hm \Km \Hm^\H) &\ge \sum_{i=1}^{N} \log (1+ \lambda_i \sigma^2  e^{\zeta}) + O(1).
  \end{align}
  The difference of the upper and lower bounds can be further bounded by
  \begin{align}
    &\log (1+\norm{\hat{\Hm}}^2 \lambda_i) - \log (1+ \lambda_i \sigma^2  e^{\zeta}) \le - \log (\sigma^2) + O(1)
  \end{align}
  where $\zeta \defeq \frac{1}{N}\sum_{i=1}^N \psi(N-i+1)$ and $\psi(x)$ is the digamma function that given by~\cite{Goodman:1963,Oyman:03}
  \begin{align}
    \psi(x) = - \gamma + \sum_{p=1}^{x-1} \frac{1}{p} \le \ln x
  \end{align}
  for integer $x$, where $\gamma$ is Euler's constant.
\end{lemma}

\begin{proof}
  Please refer to Appendix for details.
\end{proof}

Hence, we can outer-bound the weighted sum rate as
\begin{align}
  n (R_1 + 2R_2) &\le nN(2+\alpha) \log P + n \cdot O(1) + n \epsilon_n
\end{align}
and similarly obtain another outer bound by exchanging the roles of Rx-1 and Rx-2, i.e.,
\begin{align}
  n (2R_1 + R_2) &\le nN(2+\alpha) \log P + n \cdot O(1) + n \epsilon_n
\end{align}
and therefore the outer bound of the DoF is obtained by the definition.

\subsection{Achievability}
For the achievability, the vertices $(N,N\alpha)$, $(N\alpha,N)$, and $\left(\frac{N(2 + \alpha)}{3},\frac{N (2 + \alpha)}{3}\right)$ are all achievable and the achievable schemes can be simply extended from the MISO case. Here, we only show the achievable scheme for the vertex $\left(\frac{N(2 + \alpha)}{3},\frac{N (2 + \alpha)}{3}\right)$ for instance. For the achievability of $(N,N\alpha)$ and $(N\alpha,N)$, the extension is similar and straightforward. In the extended scheme, three time slots are consumed, where the transmitted signals are detailed as follows:
\subsubsection*{Slot-1}
The transmitted signals from both transmitters are given by
  \begin{subequations}
  \begin{align}
    \xv_1(1) &= \Bmatrix{\Qm_{21}(1) & \Qm_{21}^{\bot}(1) } \uv(1)\\
    \xv_2(1) &= \Bmatrix{\Qm_{12}(1) & \Qm_{12}^{\bot}(1) } \vv(1)
  \end{align}
  \end{subequations}
  where $\uv(1) \in \CC^{2N \times 1}$,  $\vv(1) \in \CC^{2N \times 1}$ are assumed to satisfy $\E(\Norm{\uv_1(1)})=\E(\Norm{\vv_1(1)})\le P$ and $\Qm_{ji}(t) \in \CC^{M \times N}$, $\Qm_{ji}^{\bot}(t) \in \CC^{M \times N}$ which satisfy
  \begin{subequations}
  \begin{align}
    \Qm_{21}(t) \subseteq \Rc\{\hat{\Hm}_{21}(t)\} \quad & \quad \Qm_{12}(t) \subseteq \Rc\{\hat{\Hm}_{12}(t)\}\\
    \Qm_{21}^{\bot}(t) \subseteq \Nc \{\hat{\Hm}_{21}(t)\} \quad  & \quad \Qm_{12}^{\bot}(t) \subseteq \Nc \{\hat{\Hm}_{12}(t)\}
  \end{align}
  \end{subequations}
  where $\Rc\{\cdot\}$ and $\Nc\{\cdot\}$ represent range and null spaces, respectively. Note that the range space is with dimension $N$ whereas the null space is with dimension $M-N$. Similarly to the MISO case, the estimation error satisfies $\E[\Norm{\Hm_{ji}(t)\Qm_{ji}^{\bot}(t)}] \sim P^{-\alpha}$.

  At both receivers, we have
  \begin{subequations}
    \begin{align}
    \yv(1) &= \Hm_{11}(1) \xv_1(1) + {\etav}_{1} + \ev(1) \\
    \zv(1) &= \Hm_{22}(1) \xv_2(1) + {\etav}_{2} + \bv(1),
    \end{align}
  \end{subequations}
  where the interference vectors overheard at both receivers are
  \begin{subequations}
    \begin{align}
    {\etav}_1 &= \Hm_{12}(1) \xv_2(1) \in \CC^{N \times 1}  \\
    &=\Hm_{12}(1) \Qm_{12}(1) \vv_1(1) + \Hm_{12}(1) \Qm_{12}^{\bot}(1) \vv_2(1)\\
    {\etav}_2 &= \Hm_{21}(1) \xv_1(1) \in \CC^{N \times 1}\\
    &=\Hm_{21}(1) \Qm_{21}(1) \uv_1(1) + \Hm_{21}(1) \Qm_{21}^{\bot}(1) \uv_2(1)
    \end{align}
  \end{subequations}
  where $\uv_1(1)$, $\uv_2(1)$, $\vv_1(1)$ and $\vv_2(1)$ are all $N \times 1$ vectors. By balancing the allocated power among those vectors, i.e.,
  \begin{align}
    \E(\norm{\uv_1(1)}^2) = P^{1-\alpha} \quad & \quad \E(\norm{\uv_2(1)}^2) = P-P^{1-\alpha}\\
    \E(\norm{\vv_1(1)}^2) = P^{1-\alpha} \quad & \quad \E(\norm{\vv_2(1)}^2) = P-P^{1-\alpha}
  \end{align}
  we approximate the total power of interference vectors as $\E(\norm{{\etav}_1}^2) \sim P^{1-\alpha}$ and $\E(\norm{{\etav}_2}^2) \sim P^{1-\alpha}$. A set of source codebooks $\{\Xc_{1i},\Xc_{2i},i=1,\cdots,N\}$ with size $(1-\alpha) \log P$ bits each are generated to represent the quantized elements of the interference vectors ${\etav}_2$ and ${\etav}_1$ at the Tx-1 and Tx-2, respectively\footnote{Here, the quantization is made on each element of the vector $\etav_i$ regardless of their mutual correlation.}. The codewords representing the elements of ${{\etav}}_2$ and ${{\etav}}_1$ are chosen uniformly from $\{\Xc_{1i}\}$ and $\{\Xc_{2i}\}$ and concatenated as $\hat{\etav}_{2}$ and $\hat{\etav}_{1}$, respectively. As stated in MISO case, the indices of $\hat{\etav}_{i}$ are encoded to $\cv_i$ using a Gaussian channel codebook and then forwarded as common messages to both receivers in the following two slots. For the sake of simplicity, we omit the channel encoding and decoding process of $\cv_i$, and therefore, it looks as if $\hat{\etav}_{i}$ itself is conveyed.

\subsubsection*{Slot-2} The objective of the slot-2 is to convey the codeword vector $\hat{\etav}_{2}$ whose information is carried on a coded common message $\cv_2$ together with a new symbol vector at the Tx-1, while only a new symbol vector is sent at the Tx-2. By omitting the encoding and decoding process of $\cv_2$, the equivalent transmitted signals are
    \begin{subequations}
      \begin{align}
        \xv_1(2) &= P^{\alpha/2} \Qm_{21}(2) \hat{\etav}_{2}  + {\Qm}_{21}^{\perp}(2) \uv(2)\\
        \xv_2(2) &= {\Qm}_{12}^{\perp}(2) \vv(2)
      \end{align}
      \end{subequations}
      where $\uv(2) \in \CC^{N \times 1}$ and $\vv(2) \in \CC^{N \times 1}$. We assume $\E(\norm{\uv(2)}^2) \le P^{\alpha}$ and $\E(\norm{\vv(2)}^2) \le P^{\alpha}$ to ensure they are recoverable. By treating $\uv(2)$ and $\vv(2)$ as noise, $N \times 1$ vector $\hat{\etav}_{2}$ is retrievable with high probability provided $N$ linearly independent equations at both receivers. After that, $\uv(2)$ and $\vv(2)$ are also recoverable from $N$ linear equations at the Rx-1 and Rx-2 by subtracting $\hat{\etav}_{2}$ from the received signals.
\subsubsection*{Slot-3} The objective of the slot-3 is the same as the slot-2 but with the exchanged roles between Tx-1 and Tx-2. The equivalent transmitted signals are given as
    \begin{subequations}
      \begin{align}
        \xv_1(3) &= {\Qm}_{21}^{\perp}(3) \uv(3)\\
        \xv_2(3) &= P^{\alpha/2} \Qm_{12}(3) \hat{\etav}_{1} + {\Qm}_{12}^{\perp}(3) \vv(3)
      \end{align}
      \end{subequations}
      where $\uv(3) \in \CC^{N \times 1}$, $\vv(3) \in \CC^{N \times 1}$ are assume to satisfy power constraint $\E(\norm{\uv(3)}^2) \le P^{\alpha}$ and $\E(\norm{\vv(3)}^2) \le P^{\alpha}$. By firstly treating $\uv(3)$ and $\vv(3)$ as noise, $N \times 1$ vector $\hat{\etav}_{1}$ can be recovered with high probability given $N$ linearly independent equations at both receivers. Similarly to the slot-2, $\uv(3)$ and $\vv(3)$ can also be recovered from the subtracted received signals.

At the end of slot-3, $N \times 1$ vectors $\hat{\etav}_{1}$ and $\hat{\etav}_{2}$ can be all recovered at both receivers, serving to cancel the overheard interference as well as to provide additional linearly independent equations for $\vv(1)$ and $\uv(1)$, respectively. With $2N$ linearly independent equations, the $2N \times 1$ vectors $\uv(1)$ and $\vv(1)$ are both recoverable with high probability at its respective receiver.

The proof of the achievable DoF pair $\left(\frac{N(2 + \alpha)}{3},\frac{N (2 + \alpha)}{3}\right)$ is similar to that in the MISO case. Take Tx-1/Rx-1 pair for example. The original channel model can be transformed to an equivalent $2N \times 2N$ point-to-point MIMO channel which conveys symbol vector $\uv(1)$ yielding $N(2-\alpha)$ DoF, and two parallel $N \times N$ MIMO channels which carry $\uv(2)$ and $\uv(3)$ respectively yielding $N\alpha$ DoF each. Hence, the total $N(2+\alpha)$ DoF is achieved within three time slots, and in turn the DoF pair $\left(\frac{N(2 + \alpha)}{3},\frac{N (2 + \alpha)}{3}\right)$ is achievable by symmetry. The detailed proof can be analogically derived according to the MISO case and hence omitted here.

\section{Conclusion}
We characterize the DoF region of the two-user MISO and certain MIMO interference channels where the transmitter has access to both delayed CSI as well as an estimate of the current CSI. In particular, these results are suited to time-correlated fading channels for which a latency-prone feedback channel provided the transmitter with the delayed samples, based on which a prediction mechanism can be applied to obtain the current imperfect CSI. Our DoF region covers a family of CSIT settings, coinciding with previously reported results for extreme situations such as pure delayed CSIT and pure current CSIT. For intermediate regimes, the DoF achieving scheme relies on the forwarding to users of a suitably quantized version of prior interference obtained under imperfect linear ZF precoding at the two transmitters.

\appendix
\subsection{Proof of Lemma~2}
  We consider Rx-1 for example to illustrate $d_1=\frac{2+\alpha}{3}$ is achievable. By symmetry, $d_2$ at the Rx-2 can be similarly obtained. Given the received signal model
  \begin{align}
    y(1) &= \hv_{11}^\H(1) \xv_1(1) + {\eta}_1 + e(1) \\
    y(2) &= \underbrace{h_{11,1}^{*}(2) {P^{\alpha/2} \hat{\eta}_2 }}_{t_{11}} + \underbrace{\hv_{11}^\H(2) \hat{\hv}_{21}^{\perp}(2) u(2)}_{t_{12}} + \underbrace{\hv_{12}^\H(2) \hat{\hv}_{12}^{\perp}(2) v(2)}_{t_{13}} + e(2) \\
    y(3) &= \underbrace{h_{12,1}^{*}(3) {P^{\alpha/2} \hat{\eta}_1 }}_{s_{11}} + \underbrace{\hv_{11}^\H(3) \hat{\hv}_{21}^{\perp}(3) u(3)}_{s_{12}} + \underbrace{\hv_{12}^\H(3) \hat{\hv}_{12}^{\perp}(3) v(3)}_{s_{13}} + e(3)
  \end{align}
  we have the achievable rate of Tx-1 and Rx-1 pair, i.e.,
  \begin{align}
    &I(\uv(1),u(2),u(3);\{y(t)\}_{t=1}^3|\Tc) \\
    &=I(\uv(1);\{y(t)\}_{t=1}^3|\Tc) + I(u(2),u(3);\{y(t)\}_{t=1}^3|\uv(1),\Tc). \label{eq:mutual-information}
  \end{align}
  For the first term, by defining
  \begin{align}
  y'(1) &= h_{12,1}^{*}(3) P^{\alpha/2} y(1) - y(3) \\
  &= h_{12,1}^{*}(3) P^{\alpha/2} \hv_{11}^\H(1) \xv_1(1) + h_{12,1}^{*}(3) P^{\alpha/2}(\Delta_1+e(1)) - s_{12} - s_{13} - e(3)\\
  y'(2) &= h_{11,1}^{*}(2) P^{\alpha/2} {\hv}_{21}^{\H}(1) \xv_1(1)) - h_{11,1}^{*}(2) P^{\alpha/2} \Delta_2 + t_{12} + t_{13} + e(2).
\end{align}
  we have
  \begin{align}
    I(\uv(1);\{y(t)\}_{t=1}^3|\Tc) &= I(\uv(1);y'(1),y'(2),y(3)|\Tc)\\
    &= I(\uv(1);y'(1),y'(2)|\Tc)
  \end{align}
where $y(3)$ is independent of $\uv(1)$.

By formulating an equivalent MIMO channel with Gaussian input $\uv(1)$ and output $\{y'(1),y'(2)\}$, i.e.,
\begin{align}
  \Bmatrix{y'(1)\\y'(2)} &= P^{\alpha/2}\underbrace{\Bmatrix{h_{12,1}^{*}(3)  \hv_{11}^\H(1) \\ h_{11,1}^{*}(2) {\hv}_{21}^{\T}(1)} [\hat{\hv}_{21}(1) \ \hat{\hv}_{21}^{\perp}(1)]}_{\Hm}  \uv(1) \\
  &+ \underbrace{\Bmatrix{h_{12,1}^{*}(3) P^{\alpha/2}(\Delta_1+e(1)) - s_{12} - s_{13} - e(3)\\- h_{11,1}^{*}(2) P^{\alpha/2} \Delta_2 + t_{12} + t_{13} + e(2)}}_{\nv}
\end{align}
where $\E(|t_{12}|^2)=\E(|s_{12}|^2) \sim P^{\alpha}$, $\E(|t_{13}|^2)=\E(|s_{13}|^2) \sim P^{0}$, $\E(|\Delta_i|^2) \sim P^{0}$, and $\E(|e(t)|^2) \sim P^0$, we have
\begin{align}
  I(\uv(1);y'(1),y'(2)|\Tc) &=I(P^{\alpha/2}\Hm \uv(1);y'(1),y'(2)|\Tc)\\
  &=h(P^{\alpha/2}\Hm \uv(1)|\Tc) - h(P^{\alpha/2}\Hm \uv(1)|y'(1),y'(2),\Tc) \\
  &=h(P^{\alpha/2}\Hm \uv(1)|\Tc) - h(\nv|y'(1),y'(2),\Tc) \label{eq:translation-app}\\
  &\ge h(P^{\alpha/2}\Hm \uv(1)|\Tc) - h(\nv) \label{eq:conditioning-app} \\
  &= \E \log \det (P^{\alpha} \Wm^{-1} \Hm \Km_1 \Hm^\H) \label{eq:term1-1}\\
  &= (2-\alpha) \log P + O(1) \label{eq:term1-2}
\end{align}
where (\ref{eq:translation-app}) holds because translation does not change differential entropy; (\ref{eq:conditioning-app}) follows the fact that conditioning reduces entropy; (\ref{eq:term1-1}) is according to the definition of differential entropy with Gaussian distributed input, where $\Km_1 \defeq \E \{\uv(1)\uv^\H(1)\}$ and $\Wm \defeq \E \{\nv \nv^\H\}$; and (\ref{eq:term1-2}) is obtained because $\Km_1 \sim \diag{\{P^{1-\alpha}, P\}}$, $\Wm \sim P^{\alpha} \Id$ and $\Hm$ does not scale as $P$.

  For the second term, we have
  \begin{align}
    &I(u(2),u(3);\{y(t)\}_{t=1}^3|\uv(1),\Tc)\\
    &=I(u(2),u(3);y(2),y(3)|\uv(1),\Tc) \label{eq:term2-mi-1}\\
    &=I(u(2),u(3);y(2),y(3)|\uv(1),\hat{\eta}_1,\hat{\eta}_2,\Tc) +  I(u(2),u(3);\hat{\eta}_1,\hat{\eta}_2|\uv(1),\Tc) \nonumber\\
    &~~~~~~~~~~~~~~~~~~~~~~~~~~~~~~~~~~~~~~~~~- I(u(2),u(3);\hat{\eta}_1,\hat{\eta}_2|\uv(1),y(2),y(3),\Tc) \label{eq:term2-mi-2}\\
    &= I(u(2),u(3);y(2),y(3)|\uv(1),\hat{\eta}_1,\hat{\eta}_2,\Tc) \label{eq:term2-mi-3}\\
    &= I(u(2),u(3); y''(2),y''(3)|\Tc) \label{eq:term2-mi-4}
  \end{align}
  where (\ref{eq:term2-mi-1}) holds because $y(1)$ is independent of $\{u(2), u(3)\}$; (\ref{eq:term2-mi-2}) is obtained by the chain rule of mutual information; (\ref{eq:term2-mi-3}) is because $u(2),u(3)$ are independent of $\hat{\eta}_1,\hat{\eta}_2$; (\ref{eq:term2-mi-4}) is due to
  \begin{align}
    y''(2) &= y(2)-h_{11,1}^{*}(2) {P^{\alpha/2} \hat{\eta}_2 } = {t_{12}} + { {t_{13}} + e(2)}\\
    y''(3) &= y(3)-h_{12,1}^{*}(3) P^{\alpha/2} \hat{\eta}_1 =  s_{12} + {s_{13} + e(3)}
  \end{align}
  are independent of $\hat{\eta}_1,\hat{\eta}_2$ and $\uv(1)$.

  By formulating an equivalent parallel channel with Gaussian input $u(2)$ (resp.~$u(3)$) and output $y'(2)$ (resp.~$y''(3)$), we have
  \begin{align}
    &I(u(2),u(3); y''(2),y''(3)|\Tc)\\
    &= I(u(2); y''(2)|\Tc) + I(u(3); y''(3)|\Tc) \label{eq:term2-1}\\
    &= \E \log\left(1+\frac{|t_{12}|^2}{\Abs{t_{13}}+|e(2)|^2} \right) + \E \log\left(1+\frac{|s_{12}|^2}{\Abs{s_{13}}+|e(3)|^2} \right) \label{eq:term2-4}\\
    &= 2\alpha \log P + O(1) \label{eq:term2-5}
  \end{align}
   where (\ref{eq:term2-1}) is due to the chain rule of mutual information and the independence between $u(2)$ (resp.~$y''(2)$) and $u(3)$ (resp.~$y''(3)$); (\ref{eq:term2-4}) is according to the definition of differential entropy with Gaussian distributed input; the last equality is due to $\E(|t_{12}|^2)=\E(|s_{12}|^2) \sim P^{\alpha}$, $\E(|t_{13}|^2)=\E(|s_{13}|^2) \sim P^{0}$, and $\E(|e(2)|^2) =\E(|e(3)|^2) \sim P^0$.

  Substituting (\ref{eq:term1-2}) and (\ref{eq:term2-5}) into (\ref{eq:mutual-information}), by the definition of DoF
  \begin{align}
    d_1 &= \frac{1}{3} \lim_{P \to \infty} \frac {I(\uv(1),u(2),u(3);\{y(t)\}_{t=1}^3|\Tc)}{\log P},
  \end{align}
  we conclude that $d_1=\frac{2+\alpha}{3}$ is achievable. By symmetry, $d_2$ can be simultaneously obtained in a similar way. This completes the proof.

\subsection{Proof of Lemma~4}
We assume without loss of generality $M\ge N$. For $M\le N$, the similar steps can be accordingly pursued.

  First for the upper bound, we have
  \begin{align}
    \E_{\tilde{\Hm}} \log \det (\Id+\Hm \Km \Hm^\H) &= \E_{\tilde{\Hm}} \log \det (\Id+\Um_{\Hm}  \mathbf{\Sigma}_{\Hm} \Vm_{\Hm}^\H \Km \Vm_{\Hm} \mathbf{\Sigma}_{\Hm} \Um_{\Hm}^\H) \label{eq:Theorem-2-1}  \\
    &=\E_{\tilde{\Hm}} \log \det (\Id+\mathbf{\Sigma}_{\Hm}^2 \Vm_{\Hm}^\H \Km \Vm_{\Hm} ) \label{eq:Theorem-2-2}\\
    &\le \E_{\tilde{\Hm}} \log \det (\Id+\lambda_{\max}(\mathbf{\Sigma}_{\Hm}^2) \Vm_{\Hm}^\H \Km \Vm_{\Hm} ) \\
    &=\sum_{i=1}^N \E_{\tilde{\Hm}} \log  (1+\lambda_{\max}(\Hm \Hm^\H) \lambda_{i}(\Vm_{\Hm}^\H \Km \Vm_{\Hm}))\\
    &\le \sum_{i=1}^N \E_{\tilde{\Hm}} \log  (1+\lambda_{\max}(\Hm \Hm^\H) \lambda_{i}) \label{eq:Theorem-2-3}\\
    &\le \sum_{i=1}^N \E_{\tilde{\Hm}} \log  (1+\Normf{\Hm} \lambda_{i}) \label{eq:Theorem-2-4}\\
    &\le \sum_{i=1}^N \E_{\tilde{\Hm}} \log  (1+(\Normf{\hat{\Hm}}+\Normf{\tilde{\Hm}}) \lambda_{i})\\
    &\le \sum_{i=1}^{N} \log  (1+(\Normf{\hat{\Hm}}+MN \sigma^2) \lambda_i) \label{eq:Theorem-2-5}\\
    &= \sum_{i=1}^{N} \log(1+\Normf{\hat{\Hm}} \lambda_i) + \sum_{i=1}^{N} \log \left( 1 + \frac{MN \sigma^2 \lambda_i}{1+\Normf{\hat{\Hm}}\lambda_i}\right)\\
    &\le \sum_{i=1}^{N} \log(1+\Normf{\hat{\Hm}} \lambda_i) + \sum_{i=1}^{N} \log \left( 1 + \frac{MN \sigma^2}{\Normf{\hat{\Hm}}}\right)\\
    &= \sum_{i=1}^{N} \log(1+\Normf{\hat{\Hm}} \lambda_i) + O(1)
  \end{align}
  where in (\ref{eq:Theorem-2-1}), $\Hm=\Um_{\Hm}  \mathbf{\Sigma}_{\Hm} \Vm_{\Hm}^\H$ with $\mathbf{\Sigma}_{\Hm} \in \CC^{N \times N}$ and $\Vm_{\Hm} \in \CC^{M \times N}$; (\ref{eq:Theorem-2-2}) comes from the equality $\det(\Id+\Am\Bm)=\det(\Id+\Bm\Am)$; (\ref{eq:Theorem-2-3}) is due to Poincare Separation Theorem~\cite{Palomar:03,Magnus:99} that $\lambda_{i}(\Vm_{\Hm}^\H \Km \Vm_{\Hm}) \le \lambda_{i} (\Km)$ for $i=1,\cdots,N$; (\ref{eq:Theorem-2-4}) is from the fact that $\norm{\Hm}_2 \le \normf{\Hm}$; (\ref{eq:Theorem-2-5}) is obtained by applying Jensen's inequality to a concave function.

  For the lower bound, we have
  \begin{align}
    \E_{\tilde{\Hm}} \log \det (\Id+\Hm \Km \Hm^\H) &\ge \E_{\tilde{\Hm}} N \log \left(1+\det(\Hm \Km \Hm^\H)^{1/N}\right) \label{eq:Theorem-2-10} \\
    &= \E_{\tilde{\Hm}} N \log \left(1+\exp\left(\frac{1}{N} \ln \det(\Hm \Km \Hm^\H)\right)\right)\\
    &\ge N \log \left(1+\exp\left(\frac{1}{N} \E_{\tilde{\Hm}} \ln \det(\Hm \Km \Hm^\H)\right)\right) \label{eq:Theorem-2-11}
  \end{align}
  where (\ref{eq:Theorem-2-10}) comes from Minkowski's inequality and (\ref{eq:Theorem-2-11}) is from Jensen's inequality by noticing that $\log(1+e^x)$ is a convex function in $x$. Hence, the expectation part can be simplified further as
  \begin{align}
    \E_{\tilde{\Hm}} \ln \det(\Hm \Km \Hm^\H) &= \E_{\tilde{\Hm}} \ln \det\left( \mathbf{\Phi} \mathbf{\Lambda}  \mathbf{\Phi}^\H \right)\\
    &\ge \E_{\tilde{\Hm}} \ln \det\left( \mathbf{\Phi}' \mathbf{\Lambda}'  \mathbf{\Phi}'^\H \right) \label{eq:Theorem-2-12}\\
    &= \ln \prod_{i=1}^N \lambda_i + \E_{\tilde{\Hm}} \ln \det\left( \mathbf{\Phi}' \mathbf{\Phi}'^\H \right) \\
    &\ge \ln \prod_{i=1}^N \lambda_i + \E_{\tilde{\Hm}} \ln \det\left( \tilde{\Hm} \Vm' \Vm'^\H \tilde{\Hm}^\H \right) \label{eq:Theorem-2-13}\\
    &= \ln \prod_{i=1}^N \lambda_i + \E_{\tilde{\Hm}} \ln \det\left( \tilde{\Hm}' \tilde{\Hm}'^\H \right) \label{eq:Theorem-2-14}\\
    &= \ln \prod_{i=1}^N (\lambda_i \sigma^2) + \sum_{i=1}^N \psi(N-i+1)
  \end{align}
  where $\mathbf{\Phi}=\Hm \Vm \in \CC^{N \times M}$ with $\Vm$ being the unitary matrix containing the eigenvectors of $\Km$, i.e., $\Km=\Vm \mathbf{\Lambda} \Vm^\H$ with $\mathbf{\Lambda}=\diag(\lambda_1,\cdots,\lambda_M)$; in (\ref{eq:Theorem-2-12}), $\mathbf{\Phi}'=\Hm \Vm'\in \CC^{N \times N}$ with $\Vm'$ being the first $N$ columns of $\Vm$ where $\mathbf{\Lambda}'=\diag(\lambda_1,\cdots,\lambda_N)$ and $\mathbf{\Phi} \mathbf{\Lambda}  \mathbf{\Phi}^\H \succeq \mathbf{\Phi}' \mathbf{\Lambda}'  \mathbf{\Phi}'^\H$; (\ref{eq:Theorem-2-13}) is due to the fact that the capacity of a Ricean fading channel is no less than that of a Rayleigh fading channel in all SNR region~\cite{Kim:03,Hosli:05}, by treating $\hat{\Hm}$ as the deterministic line-of-sight component; (\ref{eq:Theorem-2-14}) is because $\tilde{\Hm}$ is independent of $\Vm'$ and $\tilde{\Hm}'$ is an $N \times N$ i.i.d.~complex Gaussian matrix with each entry being zero-mean and variance-$\sigma^2$; the last equation is according to the isotropic assumption and obtained from~\cite{Goodman:1963,Oyman:03}.

  By the definition of $\zeta$, we have
  \begin{align}
    N \log \left(1+\exp\left(\frac{1}{N} \E_{\tilde{\Hm}} \ln \det(\Hm \Km \Hm^\H)\right)\right) &\ge N \log \left(1+\exp\left(\frac{1}{N} \ln \prod_{i=1}^N (\lambda_i \sigma^2 e^{\zeta}) \right)\right)\\
    &= N \log \left(1+\left( \prod_{i=1}^N (\lambda_i \sigma^2 e^{\zeta}) \right)^{\frac{1}{N}}\right)\\
    &\ge N \frac{1}{N} \sum_{i=1}^N \left[\log (\lambda_i \sigma^2 e^{\zeta}) \right]^{+} \\
    &\ge \sum_{i=1}^N \log (1+\lambda_i \sigma^2 e^{\zeta})  - N\\
    &= \sum_{i=1}^N \log (1+\lambda_i \sigma^2 e^{\zeta})  + O(1)
  \end{align}
  where $(x)^{+} \defeq \max\{x,0\}$ and last inequality due to the fact $(\log(x))^{+} \ge \log(1+x)-1$.

  The difference between upper and lower bounds can be further bounded as
  \begin{align}
    \log (1+\norm{\hat{\Hm}}^2 \lambda_i) - \log (1+ \lambda_i \sigma^2  e^{\zeta})  &\le \log\left(\frac{1+\norm{\hat{\Hm}}^2 \lambda_i}{1+ \lambda_i \sigma^2  e^{\zeta}}\right)\\
    &\le \log\left(1+\frac{\norm{\hat{\Hm}}^2 }{\sigma^2  e^{\zeta}}\right)\\
    &\le -\log(\sigma^2) + \log({\norm{\hat{\Hm}}^2+\sigma^2  e^{\zeta}})\\
    &\le -\log(\sigma^2) + \log({\norm{\hat{\Hm}}^2+e^{\zeta}})\\
    &= -\log(\sigma^2) + O(1)
  \end{align}
  where the second inequality is due to the fact $\log\left(\frac{1+a}{1+b}\right) \le \log\left(1+\frac{a}{b}\right)$ and the last equality is because both $\norm{\hat{\Hm}}^2$ and $ e^{\zeta}$ are bounded.
  This completes the proof.

\end{document}